\RequirePackage{fix-cm}

\documentclass{svjour3}                     
\smartqed  
\usepackage{graphicx}
\usepackage{epstopdf}
\usepackage{amsmath,amssymb}

\usepackage{xcolor}
\usepackage{xspace }

\usepackage[font=footnotesize, skip = 1pt, labelfont=bf]{caption}
\usepackage{subcaption}
\usepackage{adjustbox}

\usepackage{algorithmic}
\usepackage{algorithm}
\usepackage{wrapfig}

\usepackage{booktabs}
\usepackage{eurosym}

\newcommand{\aladin}{\textsc{aladin}\xspace}
\newcommand{\admm}{\textsc{admm}\xspace}

\newcommand{\opf}{\textsc{opf}\xspace}
\newcommand{\ac}{\textsc{ac}\xspace}

\newcommand{\ieee}{\textsc{ieee}\xspace}

\begin{document}

\title{Feasibility vs. Optimality in Distributed AC OPF---\\A Case Study Considering ADMM and ALADIN 
\thanks{{This work is part of  a project that receives funding from the European Union's Horizon 2020 research and innovation program under grant agreement No. 730936. TF acknowledges further support from the Baden-W\"urttemberg Stiftung under the Elite Programme for Postdocs.}}
}

\titlerunning{Feasibility vs. Optimality in Distributed AC OPF}        

\author{Alexander Engelmann$^\star$        \and
	Timm Faulwasser$^\star$ 
}


\institute{$^\star$ Institute for Automation and Applied Informatics (IAI), Karlsruhe Institute of Technology (KIT),  
	Hermann-von-Helmholtz-Platz 1, 76344 Eggenstein-Leopoldshafen, Germany. \\
	\email{alexander.engelmann@kit.edu, timm.faulwasser@ieee.org}           
}


\maketitle

\begin{abstract}
This paper investigates the role of feasible initial guesses and large con\-sensus-violation penalization in distributed optimization for Optimal Power Flow (\opf) problems. Specifically, we discuss the behavior of the Alternating Direction of Multipliers Method (\admm). We show that in case of large consensus-violation penalization \admm might exhibit slow progress. We support this observation by an analysis of the algorithmic properties of \admm. Furthermore, we illustrate our findings considering the \ieee 57 bus system and we draw upon a comparison of \admm  and the Augmented Lagrangian Alternating Direction Inexact Newton (\aladin) method. 
\end{abstract}

\section{Introduction}

Distributed {optimization algorithms} for \ac Optimal Power Flow (\opf) recently gained significant interest as these problems are inherently non-convex and often large-scale; i.e. comprising up to several thousand buses \cite{Guo2017}. Distributed optimization is considered to be helpful as it allows splitting large \opf problems into several smaller subproblems; thus reducing complexity and avoiding the exchange of full grid models between subsystems. We refer to \cite{Molzahn2017} for a recent overview of distributed optimization and control approaches in power systems.

One frequently discussed \emph{convex} distributed optimization method is the Alternating Direction of Multipliers Method (\admm) \cite{Boyd2011}, which is also applied in context of \ac \opf \cite{Guo2017,Erseghe2014a,Kim2000}. 
\admm often yields promising results even for large-scale power systems \cite{Erseghe2014a}. However, \admm sometimes requires a specific partitioning technique and/or a feasible initialization in combination with high consensus-violation penalization parameters to converge \cite{Guo2017}.  
The requirement of feasible initialization seems quite limiting as it requires solving a centralized inequality-constrained power flow problem requiring full topology and load information leading to approximately the same complexity as full \opf.

In previous works \cite{Engelmann2019,Engelmann18a,Murray2018} we suggested applying the Augmented Lagrangian Alternating Direction Inexact Newton (\aladin) method  to stochastic and deterministic \opf problems ranging from 5 to 300 buses.
In case a certain line-search is applied \cite{Houska2016}, \aladin provides global convergence guarantees to local minimizers for non-convex problems without the need of feasible initialization. 
The results in \cite{Engelmann2019} underpin that \aladin is able to outperform \admm  in many cases. This comes at cost of a higher  per-step  information exchange compared with \admm and a more complicated coordination step, cf. \cite{Engelmann2019}. 

In this paper we investigate the interplay of feasible initialization with high penalization for consensus violation in \admm for distributed \ac \opf.  We illustrate our findings on the \ieee 57-bus system. Furthermore, we compare the convergence behavior of \admm to  \aladin not suffering from the practical need for feasible initialization \cite{Houska2016}. 
Finally, we provide theoretical results supporting our numerical observations for the convergence behavior of \admm.

The paper is organized as follows: In Section \ref{sec:ALADINandADMM} we briefly recap \admm and \aladin including their convergence properties. Section \ref{sec:NumRes} shows numerical results for the \ieee 57-bus system focusing on the influence of the penalization parameter $\rho$ on the convergence behavior of \admm. Section \ref{sec:Theory} presents an analysis of the  interplay between high penalization and a feasible initialization.

\section{ALADIN and ADMM} \label{sec:ALADINandADMM}
{For distributed optimization, \opf problems are often formulated in \emph{affinely coupled separable form}
\begin{equation} \label{eq:OPF}
		\min_{x\in\mathbb{R}^{n_x}} \, \sum_{i\in \mathcal{R}} f_i(x_i) \quad\text{subject to}\quad x_i \in \mathcal{X}_i,\; \forall\, i \in \mathcal{R} \quad\text{and}\quad
	\sum_{i\in \mathcal{R}}A_i x_i=0,
\end{equation}
where the decision vector is divided into sub-vectors $x^\top=[x_1^\top,\dots,x_{|\mathcal{R}|}^\top] \in \mathbb{R}^{n_x}$, $\mathcal{R}$ is the index set of subsystems usually representing geographical areas of a power system and local nonlinear constraint sets $\mathcal{X}_i:=\{x_i \in \mathbb{R}^{n_{xi}} \; | \; h_{i}(x_i) \leq 0\}$. }Througout this work we assume that $f_i$ and $h_i$ are twice continuously differentialble and that all $\mathcal{X}_i$ are compact. Note that the objective functions $f_i:\mathbb{R}^{n_{xi}}\rightarrow \mathbb{R}$ and nonlinear inequality constraints $h_i: \mathbb{R}^{n_{xi}}\rightarrow \mathbb{R}^{n_{hi}}$ only depend on $x_i \in \mathbb{R}^{n_i}$ and that coupling between them takes place in the affine  consensus constraint $\sum_{i\in \mathcal{R}}A_i x_i=0$ only. There are several ways of formulating \opf problems in form of \eqref{eq:OPF} differing in the coupling variables and the type of the power flow equations (polar or rectangular), cf. \cite{Engelmann2019,Erseghe2014a,Kim1997}. 

\setlength{\textfloatsep}{3pt}
\begin{algorithm}[t]
	\caption{ADMM}
	\small
	\textbf{Initialization:} Initial guesses $z_i^0,\lambda_i^0$ for all $i \in \mathcal{R}$, parameter $\rho$. \\
	\textbf{Repeat} (until convergence):
	\begin{enumerate}
		\item \textit{Parallelizable Step:} Solve for all $i\in \mathcal{R}$ locally
		\begin{align}\label{eq:locStepADM}
		x_i^k=\underset{x_i}{\text{arg}\text{min}}&\quad f_i(x_i) + (\lambda^k_i)^\top A_i x_i + \frac{\rho}{2}\left\|A_i(x_i-z_i^k)\right\|_2^2\quad
		\text{s.t.} \quad h_i(x_i) \leq 0.		
		\end{align}
		
		\item \textit{Update dual variables} 
		\begin{equation} \label{eq:dualUp}
		\lambda_i^{k+1} = \lambda_i^{k} + \rho A_i(x_i^k-z_i^k)
		\end{equation}	
		\item \textit{Consensus Step:} Solve the coordination  problem
		\begin{subequations} \label{eq:ADMconsStep}
			\begin{align}
			&\underset{\Delta x}{\text{min}}\;\sum_{i\in \mathcal{R}} \textstyle \frac{\rho}{2}\Delta x_i^\top A_i^\top A_i\Delta x_i + \lambda_i^{k+1\top}A_i \Delta x_i \\ 
			&\begin{aligned} \label{eq:ConsCnstrADM}
			\text{s.t.}                                &\;\;\, \sum_{i\in \mathcal{R}}A_i(x^k_i+\Delta x_i) =  0     &&
			\end{aligned}
			\end{align}
		\end{subequations}
		\item \textit{Update variables} 
		$\quad  \label{eq::update} 
		z^{k+1}\leftarrow  x^k + \Delta x^k.
		$
		
	\end{enumerate}
	\label{alg:ADMM}
\end{algorithm}
\begin{algorithm}[h]
	\caption{ALADIN (full-step variant)}
	\small
	\textbf{Initialization:} Initial guess $(z^0,\lambda^0)$, parameters $\Sigma_i\succ 0,\rho,\mu$. \\
	\textbf{Repeat} (until convergence):
	\begin{enumerate}
		\item \textit{Parallelizable Step:} Solve for all $i\in \mathcal{R}$ locally
		\begin{align}\notag
		x_i^k=\underset{x_i}{\text{arg}\text{min}}&\quad f_i(x_i) + (\lambda^k)^\top A_i x_i + \frac{\rho}{2}\left\|x_i-z_i^k\right\|_{\Sigma_i}^2\quad
		\text{s.t.}\quad h_i(x_i) \leq 0 \quad \mid \kappa_i^k.		
		\end{align}
		\item \textit{Compute sensitivities:} 
		Compute Hessian approximations $B_i^k$, gradients $g_i^k$ and Jacobians of the active constraints $C_i^k$, cf. \cite{Houska2016}.
		\item \textit{Consensus Step:} Solve the coordination  problem
		\begin{align} 
		\notag
		&\underset{\Delta x,s}{\text{min}}\;\sum_{i\in \mathcal{R}}\left\{\textstyle \frac{1}{2}\Delta x_i^\top B^k_i\Delta x_i + {g_i^k}^\top \Delta x_i\right\} \hspace{-0.1em}  + \hspace{-0.1em} {\lambda^k}^\top  \hspace{-0.2em} s + \hspace{-0.1em}  \textstyle \frac{\mu}{2}\|s\|^2_2   \\ 
		&\begin{aligned} 
		\text{s.t.}                                &\;\;\, \sum_{i\in \mathcal{R}}A_i(x^k_i+\Delta x_i) =  s     &&|\, \lambda^\text{QP}\\\label{eq::conqp} 
		&\;\;\;  C^k_i \Delta x_i = 0                                    &&\forall i\in \mathcal{R}.
		\end{aligned}
		\end{align}
		\item \textit{Update variables} 
		$ \quad
		z^{k+1}\leftarrow  x^k + \Delta x^k,\;\;\; \lambda^{k+1} \leftarrow  \lambda^\mathrm{QP}.
		$

		\textbf{Similarity to ADMM:} Remove $C_i^k$ in \eqref{eq::conqp}, set $B_i^k=\rho A_i^\top A_i, g_i^k=A_i^\top \lambda^k$ and $\Sigma_i=A_i^\top A_i$, set $\mu = \infty$ (i.e. $s=0$) and use dual ascent step \eqref{eq:dualUp} for updating $\lambda^k$ in \eqref{eq::update}, cf. 
		\cite{Houska2016}.
	\end{enumerate}
	\label{alg:ALADIN}
\end{algorithm}

Here, we are interested in solving Problem \eqref{eq:OPF} via \admm and \aladin summarized in Algorithm \ref{alg:ADMM} and Algorithm \ref{alg:ALADIN} receptively.\footnote{{We remark that there exist a variety of variants of \admm, cf. \cite{Boyd2011,Bertsekas1989}. Here, we choose the formulation from \cite{Houska2016} in order to highlight similarities between \admm and \aladin.}}$^,$\footnote{Note that, due to space limitations, we describe the full-step variant of \aladin here. To obtain convergence guarantees from a remote starting point, a line-search is necessary, cf. \cite{Houska2016}.}   
At first glance it is apparent that \admm and \aladin share several features. For example, in Step 1) of both algorithms, local augmented Lagrangians subject to local nonlinear inequality constraints $h_i$ are minimized in parallel.\footnote{For notational simplicity, we only consider nonlinear inequality constraints here. Nonlinear equality constraints $g_i$ can be incorporated via a reformulation in terms of two inequality constraints, i.e. $0\leq g_i(x_i) \leq 0$. } Observe that while \admm maintains multiple local Lagrange multipliers $\lambda_i$,  \aladin considers one global Lagrange multiplier vector $\lambda$. 
In Step 2), \aladin computes sensitivities $B_i$, $g_i$ and $C_i$ (which often can directly be obtained from the local numerical solver without additional computation) and \admm updates the multiplier vectors $\lambda_i$.

In Step 3), both algorithms communicate certain information to a central entity which then solves a (usually centralized) coordination quadratic program. However, \aladin and \admm differ in the amount of exchanged information: Whereas \admm only communicates the the local primal {and dual} variables $x_i$ {and $\lambda_i$}, \aladin additionally communicates sensitivities. This is a considerable amount of extra information compared with \admm.  However, there exist methods to reduce the amount of exchanged information and bounds on the information exchange are given in \cite{Engelmann2019}. Another important difference is the computational complexity of the coordination step. In many cases, the coordination step in \admm can be reduced to an averaging step based on neighborhood communication only \cite{Boyd2011}, whereas in \aladin the coordination step involves the centralized solution of an equality constrained quadratic program.

In the last step, \admm updates the primal variables $z_i$, while \aladin additionally updates the dual variables $\lambda$.
Differences of \aladin and \admm also show up in the convergence speed and their theoretical convergence guarantees: Whereas \aladin guarantees global convergence  and quadratic local convergence for non-convex problems if a certain line-search is applied  \cite{Houska2016}, few results exist for \admm in the non-convex setting.
Recent works \cite{Wang2019,Hong2016} investigate the convergence of \admm for special classes of non-convex problems; however, to the best of our knowledge \opf problems do not belong to these classes.

\section{Numerical Results} \label{sec:NumRes}
Next, we investigate the behavior of \admm for large $\rho$ and a feasible initialization to illustrate performance differences between \admm and \aladin.
We consider power flow equations in polar form with coupling in active/reactive power and voltage angle and magnitude at the boundary between two neighbored regions \cite{Engelmann2019}. 
We consider the \ieee 57-bus system with data from \textsc{matpower} and partitioning as in \cite{Engelmann2019} as numerical test case. 

Figures \ref{fig:ADMAL1e4} and \ref{fig:ADMAL1e6} show the convergence behavior of \admm (with and without feasible initialization (f.)) and \aladin   for several convergence indicators and two different penalty parameters $\rho = 10^4$ and $\rho = 10^6$.\footnote{The scaling matrices $\Sigma_i$ are diagonal. They scale to improve convergence. Hence, entries corresponding to voltages and phase angles are $100$, entries corresponding to powers are $1$.}
Therein, the left-handed plot depicts the consensus gap $\|Ax^k\|_\infty$  representing the maximum mismatch of coupling variables (active/reactive power and voltage magnitude/angle) at borders between two neighbored regions. The second plot shows the objective function  value $f^k$ and the third plot presents the distance to the minimizer $\|x^k-x^\star\|_\infty$ over the iteration index $k$. The right-handed figure shows the nonlinear constraint violation $\|g(z^k)\|_\infty$  after the consensus steps \eqref{eq:ADMconsStep} and \eqref{eq::conqp} of Algorithm \ref{alg:ADMM} and \ref{alg:ALADIN} respectively  representing the maximum violation of the power flow equations.\footnote{The power flow equations for the \ieee 57-bus systems are considered as nonlinear equality constraints $g_i(x_i)=0$. Hence, $g_i(x_i)\neq0$ represents a violation of the power flow equations.}

\begin{figure}
	\centering
	{\includegraphics[trim={2.2cm 0 1.9cm 0},clip,width=1.\textwidth]{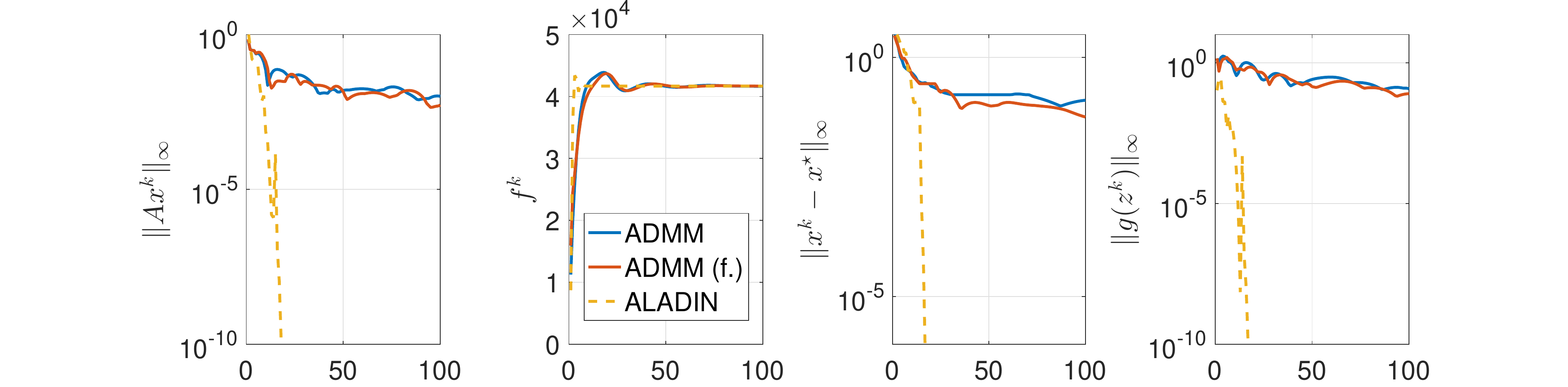}}
	\caption{Convergence behavior of \admm with infeasible initialization, \admm with feasible intialization (f.) for $\rho = 10^{4}$ and \aladin.}
	\vspace{-1em}
	\label{fig:ADMAL1e4}
\end{figure}
\begin{figure}
	\centering
	{\includegraphics[trim={2.2cm 0 1.9cm 0},clip,width=1.\textwidth]{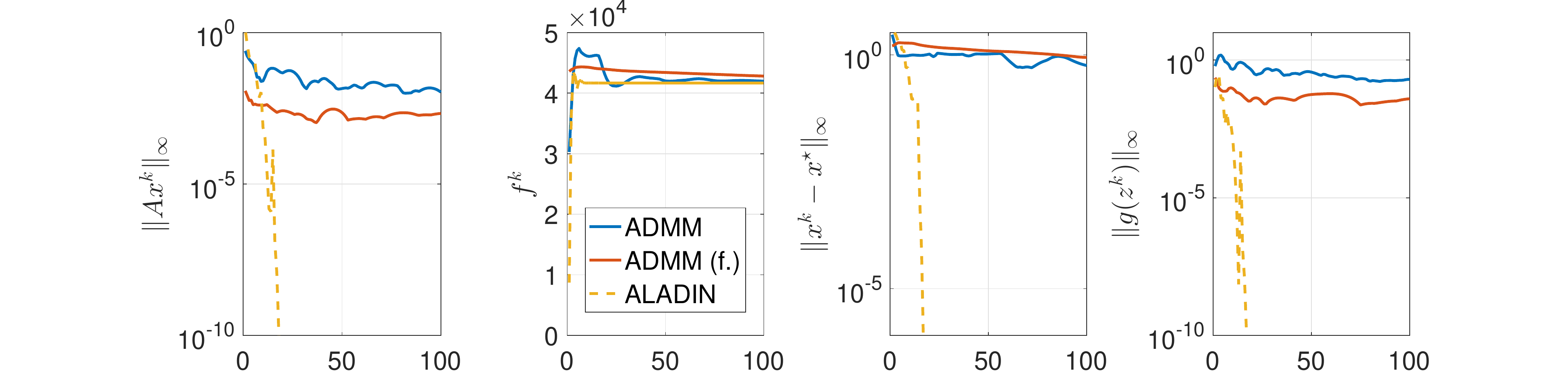}}
	\caption{Convergence behavior of \admm with infeasible initialization, \admm with feasible intialization (f.) for $\rho = 10^{6}$ and \aladin.}
	\label{fig:ADMAL1e6}
\end{figure}

In case of small $\rho = 10^4$, both \admm variants behave similar and converge slowly towards the optimal solution with slow decrease in consensus violation, nonlinear constraint violation and objective function value. 
If we increase $\rho$ to $\rho=10^{6}$ with results shown in Figure \ref{fig:ADMAL1e6}, the consensus violation $\|Ax^k\|_\infty$ gets smaller in \admm with feasible initialization. 
The reason is that a large $\rho$  forces $x_i^k$ being close to $z_i^k$ leading to small $\|Ax^k\|_\infty$ as we have $Az^k=0$ from the consensus step. 
But, at the same time, this also leads to a slower progress in optimality $f^k$ compared to $\rho =10^4$, cf. the second plot in Figures \ref{fig:ADMAL1e4} and \ref{fig:ADMAL1e6}.

{On the other hand, this statement does not hold for \admm with infeasible initialization (blue lines in Figures \ref{fig:ADMAL1e4} and \ref{fig:ADMAL1e6}) as the constraints in the local step and the consensus step of \admm enforce an alternating projection between the consensus constraint and the local nonlinear constraints. 
The progress in the nonlinear constraint violation $\|g(z^k)\|_\infty$ supports this statement. In its extreme, this behavior can be observed  when using $\rho = 10^{12}$ depicted in Figure \ref{fig:ADMAL1e12}. There, \admm with feasible initialization produces very small consensus violations and nonlinear constraint violations at cost of almost no progress in terms of optimality. }

Here the crucial observation is that in case of feasible initialization and large penalization parameter $\rho$ \admm produces \emph{almost feasible} iterates at cost of slow progress in the objective function values. 
From this, one is tempted to conclude that also for infeasible initializations \admm will likely converge, cf. Figure \ref{fig:ADMAL1e4}-\ref{fig:ADMAL1e12}. This conclusion is supported by the rather small 57-bus test system. However, it deserves to be noted that this conclusion is in general not valid, cf.  \cite{Guo2017}.

For \aladin, we use $\rho = 10^6$ and $\mu = 10^7$. 
Comparing the results of \admm with \aladin, \aladin shows superior quadratic convergence also in case of infeasible initialization. 
This is inline with the known convergence properties of \aladin \cite{Houska2016}. 
However, the fast convergence comes at the cost of increased communication overhead per step, cf. \cite{Engelmann2019} for a more detailed discussion.  Note that \aladin involves a more complex coordination step, which is not straightforward to solve via neighborhood communication. Furthermore, tuning of $\rho$ and $\mu$ can be difficult.

In power systems, usually feasibility is  preferred over optimality to ensure a stable system operation. Hence, \admm with feasible initialization and large $\rho$ can in principle be used for \opf as in this case violation of power flow equations and generator bounds are expected to be small and one could at least expect certain progress towards an optimal solution. Following this reasoning several papers consider $\|A(x^k-z^k)\|_\infty <\epsilon$  as termination criterion \cite{Erseghe2014a,Guo2017}. However, as shown in the example above, if $\rho$ is large enough, and \admm is initialized at a feasible point, this termination criterion can always be satisfied in just one iteration if $\rho$ is chosen sufficiently large. Consequently, it is unclear how to ensure a certain degree of optimality. 
An additional question with respect to \admm is how to obtain a feasible initialization. To compute such an initial guess, one  has to solve a constrained nonlinear least squares problem solving the power flow equations subject to box constraints. This is itself a problem of almost the same complexity as the full \opf problem. Hence one would again require a computationally powerful central entity with full topology and parameter information. Arguably this jeopardizes the initial motivation for using distributed optimization methods.

\begin{figure}
	\centering
	{\includegraphics[trim={2.2cm 0 1.9cm 0},clip,width=1\textwidth]{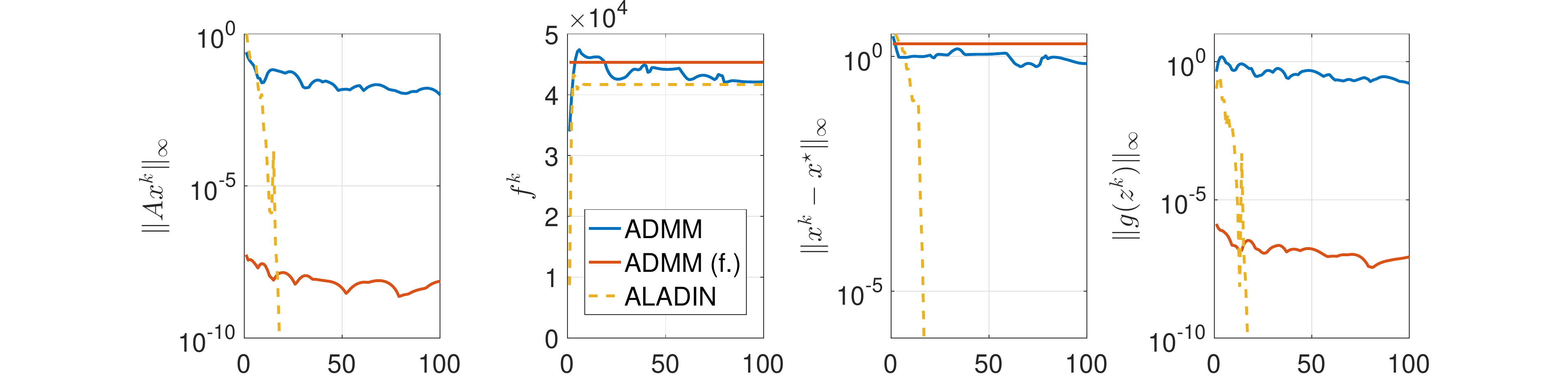}}
	\caption{Convergence behavior of \admm with infeasible initialization, \admm with feasible intialization (f.) for $\rho = 10^{12}$ and \aladin.}
	\label{fig:ADMAL1e12}
\end{figure}

\section{Analysis of ADMM with feasible initialization for $\rho \rightarrow \infty$} \label{sec:Theory}
The results above indicate that large penalization parameters $\rho$ in combination with feasible initialization might lead to pre-mature convergence to a suboptimal solution. 
Arguably, this behavior of \admm might be straight-forward to see from an optimization perspective. However, to the best of our knowledge a mathematically rigorous analysis, which is very relevant for \opf, is not available in the literature. 

{\begin{proposition}[Feasibility and $\rho \rightarrow \infty$ imply $x_i^{k+1}- x_i^k\in\operatorname{null}(A_i)$] \label{lem:ADMMfeas}\\
Consider the application of \admm (Algorithm \ref{alg:ADMM}) to Problem \eqref{eq:OPF}.
Suppose that, for all $k \in \mathbb{N}$, the local problems \eqref{eq:locStepADM} have unique regular minimizers $x_i^k$.\footnote{A minimizer is regular if the gradients of the active constraints are linear independent \cite{Nocedal2006}.}  For $\tilde k\in \mathbb{N}$, let $\lambda^{\tilde k}_i$ be bounded and,  for all $i\in \mathcal{R}$, $z_i^{\tilde k} \in \mathcal{X}_i$.
Then, the \admm iterates satisfy 
\[
\lim_{\rho \to \infty}x_i^{k}(\rho) - x_i^{\tilde k} \in \operatorname{null}(A_i),\quad \forall k>\tilde k.
\]
\end{proposition}
	\begin{proof}
The proof is divided into four steps.	Steps 1)-3) establish technical properties used to derive the above assertion in Step 4).
			
	\underline{Step 1).} 	
	At iteration $\tilde k$ the local steps of \admm are
		\begin{align}\label{eq:ProofMin}
		x_i^{\tilde k}(\rho) = \underset{x_i\in \mathcal{X}_i}{\operatorname{argmin}}& \;f_i(x_i) + \left(\lambda_i^{\tilde k}\right)^\top A_i x_i + \frac{\rho}{2}\left\|A_i\left(x_i-z_i^{\tilde k}\right)\right\|^2_2.
		\end{align}
		Now, by assumption all $f_i$ is twice continuously differentiable (hence bounded on $\mathcal{X}_i$), $\lambda_i^{\tilde k}$ is bounded and all ${z_i^{\tilde k}\in \mathcal{X}_i}$. Thus,  for all $i \in \mathcal{R}$, $\underset{\rho \rightarrow \infty}{\lim} x_i^{\tilde k}(\rho) = z_i^{\tilde k} + v_i^{\tilde k}$ 		with $v_i^{\tilde k} \in \operatorname{null}(A_i)$.

\underline{Step 2).}
  The first-order stationarity condition of \eqref{eq:ProofMin} can be written as
	\begin{align}\label{eq:statCond}
		-\nabla f_i(x_i^{\tilde k}) - \gamma_i ^{\tilde k\top} \nabla h_i( x_i^{\tilde k}) = A_i^\top \lambda_i^{\tilde k} + \rho A_i^\top A_i\left(x_i^{\tilde k}-z_i^{\tilde k}\right),
	\end{align}
		where $\gamma_i ^{\tilde k\top}$ is the multiplier associated to $h_i$. Multiplying the multiplier update formula \eqref{eq:dualUp}  with $A_i^\top$ from the left we obtain $A_i^\top\lambda_i^{k+1} = A_i^\top\lambda_i^{k} + \rho A_i^\top A_i(x_i^k-z_i^k)$. Combined with \eqref{eq:statCond} this yields
	$		A_i^\top \lambda_i^{\tilde k+1} = - \nabla f(x_i^{\tilde k}) - \gamma ^{\tilde k\top} \nabla h_i(x_i^{\tilde k})$.
		By differentiability  of $f_i$ and $h_i$, compactness of $\mathcal{X}_i$ and regularity of $x_i^{\tilde k}$ this implies boundedness of $A_i^\top \lambda_i^{\tilde k+1}$.

\underline{Step 3).} Next, we show by contradiction that $\Delta x_i^{\tilde k} \in \operatorname{null}(A_i)$ for all $i \in \mathcal{R}$ and $\rho \rightarrow \infty$.  Recall the coordination step \eqref{eq:ConsCnstrADM} in \admm given by 
			\begin{align} \label{eq:proofQP}
			\underset{\Delta x}{\text{min}}\;\sum_{i\in \mathcal{R}} \dfrac{\rho}{2}\Delta x_i^\top A_i^\top A_i\Delta x_i + \lambda_i^{\tilde k+1\top}A_i \Delta x_i 
			\;\;\;\text{s.t.} \;                               \; \sum_{i\in \mathcal{R}}A_i(x^{\tilde k}_i+\Delta x_i) =  0. 
			\end{align}
		Observe that any $\Delta x_i^{\tilde k} \in \operatorname{null}(A_i)$ is a feasible point to \eqref{eq:proofQP} as $\sum_{i\in \mathcal{R}}A_ix^{\tilde k}_i=0$.  Consider a feasible candidate solution  $\Delta x_i \notin \operatorname{null}(A_i)$ for which $ \sum_{i\in \mathcal{R}}A_i(x^{\tilde k}_i+\Delta x_i) =  0$. Clearly, $\lambda_i^{\tilde k+1\top}A_i \Delta x_i(\rho)$ will be bounded. Hence for a sufficiently large value of $\rho$, the objective of \eqref{eq:proofQP} will be positive. However, for any $\Delta x_i  \in \operatorname{null}(A_i)$ the objective of  \eqref{eq:proofQP}  is zero, which contradicts optimality of the  candidate solution  $\Delta x_i \notin \operatorname{null}(A_i)$. Hence, choosing $\rho$ sufficiently large ensures that any minimizer of \eqref{eq:proofQP} lies in $\operatorname{null}(A_i)$.
	
	\underline{Step 4).} It remains to show  $x_i^{\tilde k + 1} = x_i^{\tilde k}$. In the last step of \admm we have $z^{\tilde k+1}=x^{\tilde k} + \Delta x^{\tilde k}$. Given Steps 1-3) this yields $z^{\tilde k+1}=z^{\tilde k} + v^{\tilde k} + \Delta x^{\tilde k}$ and hence
	\[ \left\|A_i\left(x_i-z_i^{ \tilde k+1}\right)\right\|^2_2=\left\|A_i\left(x_i-z_i^{\tilde k} + v_i^{\tilde k} + \Delta x_i^{\tilde k}\right)\right\|^2_2=\left\|A_i\left(x_i-z_i^{\tilde k}\right)\right\|^2_2.\]
	Observe that this implies that, for $\rho \rightarrow \infty$, problem $\eqref{eq:ProofMin}$ does not change from step $\tilde k$ to $\tilde k+1$. This proves the assertion. $\hfill \blacksquare$
	\end{proof}}
\begin{corollary}[Deterministic code, feasibility, $\rho \rightarrow \infty$ implies $x_i^{k+1}= x_i^k$]\,\\
Assuming that the local subproblems in \admm are solved deterministically; i.e. same problem data yields the same solution.
Then under the conditions of  Proposition \ref{lem:ADMMfeas} and for $\rho \rightarrow \infty$, once \admm generates a feasible point $x_i^{\tilde k}$ to Problem \eqref{eq:OPF}, or whenever it is initialized at a feasible point, it will stay at this point  for all subsequent $k>\tilde k$.
\end{corollary}

The above corollary explains the behavior of \admm for large $\rho$ in combination with feasible initialization often used in power systems \cite{Guo2017,Erseghe2014a}. 
Despite feasible iterates are desirable from a power systems point of view, the findings above imply that high values of $\rho$ limit progress in terms of minimizing the objective. 

\begin{remark}[Behavior of \aladin for $\rho \to \infty$]\,\\
Note that for $\rho \to \infty$, \aladin behaves different than \admm.  While the local problems in \aladin  behave similar to \admm, the coordination step in \aladin is equivalent to a sequential quadratic programming step. This helps avoiding premature convergence and it ensures decrease of $f$ in the coordination step \cite{Houska2016}.
\end{remark}

\section{Conclusions}
This method-oriented work investigated the interplay of penalization of consensus violation and feasible initialization in \admm. We found that---despite often working reasonably with a \emph{good choice} of $\rho$ and infeasible initialization---in case of feasible initialization combined with large values of $\rho$ \admm typically stays feasible yet it may stall at a suboptimal solution.  We provided analytical results supporting this observation. 
However, computing a feasible initialization is itself a problem of almost the same complexity as the full \opf problem; in some sense partially jeopardizing the advantages of distributed optimization methods. 
Thus distributed methods providing rigorous convergence guarantees while allowing for infeasible initialization are of interest. 
One such alternative method is \aladin \cite{Houska2016} exhibiting convergence properties at cost of an enlarged communication overhead and a more complex coordination step \cite{Engelmann2019}.


%
\bibliographystyle{abbrv}


\bibliography{paper}

\begin{thebibliography}{10}

\bibitem{Bertsekas1989}
D.~P. Bertsekas and J.~N. Tsitsiklis.
\newblock {\em Parallel and Distributed Computation: Numerical Methods},
  volume~23.
\newblock Prentice Hall Englewood Cliffs, NJ, 1989.

\bibitem{Boyd2011}
S.~Boyd, N.~Parikh, E.~Chu, B.~Peleato, and J.~Eckstein.
\newblock Distributed optimization and statistical learning via the alternating
  direction method of multipliers.
\newblock {\em Found Trends Mach Learn}, 3(1):1--122, 2011.

\bibitem{Engelmann2019}
A.~Engelmann, Y.~Jiang, T.~M{\"u}hlpfordt, B.~Houska, and T.~Faulwasser.
\newblock Toward distributed {OPF} using {ALADIN}.
\newblock {\em IEEE Trans Power Syst}, 34(1):584--594, 2019.

\bibitem{Engelmann18a}
A.~Engelmann, T.~M\"uhlpfordt, Y.~Jiang, B.~Houska, and T.~Faulwasser.
\newblock Distributed stochastic {AC} optimal power flow based on polynomial
  chaos expansion.
\newblock In {\em Proc Am Control Conf (ACC) 2018}, pages 6188--6193, June
  2018.

\bibitem{Erseghe2014a}
T.~Erseghe.
\newblock Distributed optimal power flow using {ADMM}.
\newblock {\em IEEE Trans Power Syst}, 29(5):2370--2380, Sept. 2014.

\bibitem{Guo2017}
J.~Guo, G.~Hug, and O.~K. Tonguz.
\newblock A case for nonconvex distributed optimization in large-scale power
  systems.
\newblock {\em IEEE Trans Power Syst}, 32(5):3842--3851, Sept. 2017.

\bibitem{Hong2016}
M.~Hong, Z.-Q. Luo, and M.~Razaviyayn.
\newblock Convergence analysis of alternating direction method of multipliers
  for a family of nonconvex problems.
\newblock {\em SIAM J Optim}, 26(1):337--364, 2016.

\bibitem{Houska2016}
B.~Houska, J.~Frasch, and M.~Diehl.
\newblock An augmented {L}agrangian based algorithm for distributed {nonconvex}
  optimization.
\newblock {\em SIAM J Optim}, 26(2):1101--1127, 2016.

\bibitem{Kim1997}
B.~H. Kim and R.~Baldick.
\newblock Coarse-grained distributed optimal power flow.
\newblock {\em IEEE Trans Power Syst}, 12(2):932--939, May 1997.

\bibitem{Kim2000}
B.~H. Kim and R.~Baldick.
\newblock A comparison of distributed optimal power flow algorithms.
\newblock {\em IEEE Trans Power Syst}, 15(2):599--604, May 2000.

\bibitem{Molzahn2017}
D.~K. Molzahn, F.~D\"orfler, H.~Sandberg, S.~H. Low, S.~Chakrabarti,
  R.~Baldick, and J.~Lavaei.
\newblock A survey of distributed optimization and control algorithms for
  electric power systems.
\newblock {\em IEEE Trans Smart Grid}, 8(6):2941--2962, Nov. 2017.

\bibitem{Murray2018}
A.~Murray, A.~Engelmann, V.~Hagenmeyer, and T.~Faulwasser.
\newblock Hierarchical distributed mixed-integer optimization for reactive
  power dispatch.
\newblock {\em IFAC-PapersOnLine}, 51(28):368 -- 373, 2018.
\newblock 10th IFAC Symposium on Control of Power and Energy Systems CPES 2018.

\bibitem{Nocedal2006}
J.~Nocedal and S.~Wright.
\newblock {\em Numerical Optimization}.
\newblock Springer Science \& Business Media, New York, 2006.

\bibitem{Wang2019}
Y.~Wang, W.~Yin, and J.~Zeng.
\newblock Global convergence of {ADMM} in nonconvex nonsmooth optimization.
\newblock {\em J Sci Comput}, 78(1):29--63, Jan 2019.

\end{thebibliography}

\newpage

\end{document}